%% file: main.tex
\documentclass[11pt]{article}

\usepackage[margin=1in]{geometry}
\usepackage{amsthm,amssymb,amsmath}  
\usepackage{xspace,enumerate}
\usepackage[dvipsnames]{xcolor}
\usepackage[colorlinks=true,urlcolor=Blue,citecolor=Green,linkcolor=BrickRed]{hyperref}
\usepackage[capitalise]{cleveref}
\usepackage[utf8]{inputenc}
\usepackage{thmtools}
\usepackage{thm-restate}
\usepackage{authblk}

\title{Constructing Antidictionaries in Output-Sensitive Space}
\author[1]{Lorraine A.K. Ayad}
\author[2]{Golnaz Badkobeh}
\author[3]{Gabriele Fici}
\author[4]{Alice H\'eliou}
\author[5]{Solon P. Pissis}

\affil[1]{Department of Informatics, King's College London, London, UK\\
\texttt{lorraine.ayad@kcl.ac.uk}}
\affil[2]{Department of Computing, Goldsmiths University of London, London, UK\\
\texttt{g.badkobeh@gold.ac.uk}}
\affil[3]{Dipartimento di Matematica e Informatica, Universit\`{a} di Palermo, Palermo, Italy\\
\texttt{gabriele.fici@unipa.it}}
\affil[4]{Independent Researcher\\ 
\texttt{alice.heliou@gmail.com}}
\affil[5]{CWI, Amsterdam, The Netherlands\\ 
\texttt{solon.pissis@cwi.nl}}

\usepackage{microtype}
\usepackage{amsmath,amsfonts}
\usepackage[cmbtt]{bold-extra}
\usepackage[T1]{fontenc}
\usepackage[utf8]{inputenc}
\usepackage{thm-restate}
\usepackage{comment}
\usepackage{forest}
\usepackage{xspace}
\usepackage{enumerate}
\usepackage{makecell}
\usepackage{listings}
\usepackage{multirow}
\usepackage{diagbox}
\usepackage{todonotes}
\usepackage{thmtools}
\usepackage[ruled,noline,noend]{algorithm2e}
\usepackage[capitalise]{cleveref}
\usepackage{graphics,adjustbox}
\usepackage{tabularx}
\usepackage{amssymb}
\usepackage{epsfig}
\usepackage{verbatim}
\usepackage{braket}
\usepackage[noadjust]{cite}
\usepackage{xspace}
\usepackage{bold-extra}
\usepackage[margin=1in]{geometry}

\theoremstyle{plain}
\newtheorem{theorem}{Theorem}
\newtheorem{lemma}{Lemma}  
\newtheorem{corollary}[theorem]{Corollary}  
\newtheorem{fact}{Fact}

\theoremstyle{definition}
\newtheorem{definition}{Definition}

\newtheorem{example}{Example}

\def\dd{\mathinner{.\,.}}
\newcommand{\cO}{\mathcal{O}}

\newcommand{\MAW}{\text{MAW}}
\newcommand{\MAWs}{\text{MAWs}}
\DeclareMathOperator{\Mina}{\ensuremath{\mathrm{M}}}
\DeclareMathOperator{\RMina}{\ensuremath{\mathrm{R}}}

  \def\dd{\mathinner{.\,.}}

\begin{document}
\date{}
\maketitle

\begin{abstract}
A word $x$ that is absent from a word $y$ is called {\em minimal} if all its proper factors occur in $y$. Given a collection of $k$ words $y_1,y_2,\ldots,y_k$ over an alphabet $\Sigma$, we are asked to compute the set $\Mina^{\ell}_{y_{1}\#\ldots\#y_{k}}$ of minimal absent words of length at most $\ell$ of word $y=y_1\#y_2\#\ldots\#y_k$, $\#\notin\Sigma$. In data compression, this corresponds to computing the antidictionary of $k$ documents. In bioinformatics, it corresponds to computing words that are absent from a genome of $k$ chromosomes. This computation generally requires $\Omega(n)$ space for $n=|y|$ using any of the plenty available $\cO(n)$-time algorithms. This is because an $\Omega(n)$-sized text index is constructed over $y$ which can be impractical for large $n$. We do the identical computation incrementally using output-sensitive space. This goal is reasonable when $||\Mina^{\ell}_{y_{1}\#\ldots\#y_{N}}||=o(n)$, for all $N\in[1,k]$. For instance, in the human genome, $n \approx 3\times 10^9$ but $||\Mina^{12}_{y_{1}\#\ldots\#y_{k}}|| \approx 10^6$. 
We consider a constant-sized alphabet for stating our results. 
We show that {\em all}  $\Mina^{\ell}_{y_{1}},\ldots,\Mina^{\ell}_{y_{1}\#\ldots\#y_{k}}$ can be computed in $\cO(kn+\sum^{k}_{N=1}||\Mina^{\ell}_{y_{1}\#\ldots\#y_{N}}||)$ total time using $\cO(\textsc{MaxIn}+\textsc{MaxOut})$ space, where $\textsc{MaxIn}$ is the length of the longest word in $\{y_1,\ldots,y_{k}\}$ and $\textsc{MaxOut}=\max\{||\Mina^{\ell}_{y_{1}\#\ldots\#y_{N}}||:N\in[1,k]\}$.
Proof-of-concept experimental results are also provided confirming our theoretical findings and justifying our contribution.
\end{abstract}


\section{Introduction}

The word $x$ is an \textit{absent word} of the word
$y$ if it does not occur in $y$. The absent word $x$ of $y$ is called \textit{minimal} if and only if all its proper factors occur in $y$. The set of all minimal absent words for a word $y$ is denoted by $\Mina_y$. The set of all minimal absent words of length at most $\ell$ of a word $y$ is denoted by $\Mina^\ell_y$.  For example, if $y=\texttt{abaab}$, then $\Mina_y=\{\texttt{aaa}, \texttt{aaba}, \texttt{bab}, \texttt{bb}\}$ and $\Mina^3_y=\{\texttt{aaa},\texttt{bab}, \texttt{bb}\}$. 
The upper bound on the number of minimal absent words is $\cO(\sigma n)$~\cite{DBLP:journals/ipl/CrochemoreMR98}, where $\sigma$ is the size of the alphabet and $n$ is the length of $y$, and this is tight for integer alphabets~\cite{Charalampopoulos2018}; in fact, for large alphabets, such as when $\sigma\geq\sqrt{n}$, this bound is also tight even for minimal absent words having the same length~\cite{DBLP:journals/almob/AlmirantisCGIMP17}. 

 State-of-the-art algorithms compute all minimal absent words of $y$ in $\cO(\sigma n)$ time~\cite{DBLP:journals/ipl/CrochemoreMR98,DBLP:journals/bmcbi/BartonHMP14} or in $\cO(n + |\Mina_y|)$ time~\cite{DBLP:conf/mfcs/FujishigeTIBT16,SPIRE2018} for integer alphabets. There also exist space-efficient data structures based on the Burrows-Wheeler transform of $y$ that can be applied for this computation~\cite{DBLP:conf/esa/BelazzouguiCKM13,DBLP:journals/algorithmica/BelazzouguiC17}. In many real-world applications of minimal absent words, such as in data compression~\cite{DCA,DBLP:conf/sccc/CrochemoreN02,DBLP:conf/dcc/FialaH08,DBLP:conf/isita/OtaM10}, in sequence comparison~\cite{Charalampopoulos2018,SPIRE2018}, in on-line pattern matching~\cite{DBLP:conf/fct/CrochemoreHKMPR17}, or in identifying pathogen-specific signatures~\cite{DBLP:journals/bioinformatics/SilvaPCPF15}, only a subset of minimal absent words may be considered, and, in particular, the minimal absent words of length (at most) $\ell$.
 Since, in the worst case, the number of minimal absent words of $y$ is $\Theta(\sigma n)$, $\Omega(\sigma n)$ space is required to represent them explicitly. In~\cite{SPIRE2018}, the authors presented an $\cO(n)$-sized data structure for outputting minimal absent words of a specific length in optimal time for integer alphabets.

The problem with existing algorithms for computing minimal absent words is that they make use of $\Omega(n)$ space; and the same amount is required even if one is merely interested in the minimal absent words of length at most $\ell$. This is because all of these algorithms construct global data structures, such as the suffix array~\cite{DBLP:journals/bmcbi/BartonHMP14}. In theory, this problem can be addressed by using the external memory algorithm for computing minimal absent words presented in~\cite{DBLP:journals/bioinformatics/HeliouPP17}. The I/O-optimal version of this algorithm, however, requires a lot of external memory to build the global data structures for the input~\cite{DBLP:conf/cpm/KarkkainenKP15a}. One could also use the algorithm of~\cite{DBLP:conf/spire/FujishigeTH18} that computes $\Mina^\ell_y$ in $\cO(n+|\Mina^\ell_y|)$ time using $\cO(\min\{n,\ell z\})$ space, where $z$ is the size of the LZ77 factorisation of $y$. This algorithm also requires constructing the truncated DAWG, a type of global data structure which could take space $\Omega(n)$. Thus, in this paper, we investigate whether $\Mina^\ell_y$ can be computed efficiently in output-sensitive space. As $y$ can be ``decomposed'' into a collection of $k$ words---with a suitable overlap of length $\ell$ so as not to lose information---we consider the following, general, computational problem.

\paragraph{Problem} Given $k$ words $y_1,y_2,\ldots,y_k$ over an alphabet $\Sigma$ and an integer $\ell>0$, compute the set $\Mina^{\ell}_{y_{1}\#\ldots\#y_{k}}$ of minimal absent words of length at most $\ell$ of $y=y_1\#y_2\#\ldots\#y_k$, $\#\notin\Sigma$.\\

In data compression, this scenario corresponds to computing the antidictionary of $k$ documents~\cite{DCA, DBLP:conf/sccc/CrochemoreN02}. In bioinformatics, it corresponds to computing words that are absent from a genome of $k$ chromosomes. As discussed above, this computation generally requires $\Omega(n)$ space for $n=|y|$. We do the identical computation incrementally using output-sensitive space. This goal is reasonable when $||\Mina^{\ell}_{y_{1}\#\ldots\#y_{N}}||=o(n)$, for all $N\in[1,k]$. In the human genome, $n \approx 3\times 10^9$ but $||\Mina^{12}_{y_{1}\#\ldots\#y_{k}}|| \approx 10^6$, where $k$ is the total number of chromosomes. 

\paragraph{Our Results} 
Antidictionary-based compressors work on $\Sigma=\{\texttt{0},\texttt{1}\}$ and in bioinformatics we have $\Sigma=\{\texttt{A},\texttt{C},\texttt{G},\texttt{T}\}$; we thus consider a constant-sized alphabet for stating our results. 
We show that {\em all}  $\Mina^{\ell}_{y_{1}},\ldots,\Mina^{\ell}_{y_{1}\#\ldots\#y_{k}}$ can be computed in $\cO(kn+\sum^{k}_{N=1}||\Mina^{\ell}_{y_{1}\#\ldots\#y_{N}}||)$ total time using $\cO(\textsc{MaxIn}+\textsc{MaxOut})$ space, where $\textsc{MaxIn}$ is the length of the longest word in $\{y_1,\ldots,y_{k}\}$ and $\textsc{MaxOut}=\max\{||\Mina^{\ell}_{y_{1}\#\ldots\#y_{N}}||:N\in[1,k]\}$.
Proof-of-concept experimental results are provided confirming our theoretical findings and justifying our contribution.
\section{Preliminaries}

We generally follow~\cite{DBLP:books/daglib/0020103}. An {\em alphabet} $\Sigma$ is a finite ordered non-empty set of elements called {\em letters}. A {\em word} is a sequence of elements of $\Sigma$. 
The set of all words over $\Sigma$ of length at most $\ell$ is denoted by $\Sigma^{\leq\ell}$. We fix a constant-sized alphabet $\Sigma$, i.e., $|\Sigma|=\cO(1)$. Given a word $y=uxv$ over $\Sigma$, we say that $u$ is a {\em prefix} of $y$, $x$ is a {\em factor} (or subword) of $y$, and $v$ is a {\em suffix} of $y$. We also say that $y$ is a {\em superword} of $x$. A factor $x$ of $y$ is called {\em proper} if $x\neq y$. 

Given a word $y$ over $\Sigma$, the set of \emph{minimal absent words} ($\MAWs$) of $y$ is defined as
\begin{align*}
        \Mina_y = &\{aub \mid a,b\in\Sigma, \mbox{$au$ and $ub$ are factors of $y$ but $aub$ is not}\}\\ & \cup\{c\in \Sigma \mid \mbox{$c$ does not occur in $y$}\}.
\end{align*}

For instance, over $\Sigma=\{\texttt{a,b,c}\}$, for $y=\texttt{ab}$ we have $\Mina_y=\{\texttt{aa,bb,ba,c}\}$. MAWs of length 1 for $y$ can be found in $\cO(|y|+|\Sigma|)=\cO(|y|)$ time using $\cO(|\Sigma|)=\cO(1)$ working space, and so, in what follows, we focus on the computation of MAWs of length at least 2.

The \textit{suffix tree} $\mathcal{T}(y)$ of a non-empty word $y$ of length $n$ is the compact trie representing all suffixes of $y$~\cite{DBLP:books/daglib/0020103}. The \textit{branching} nodes of the trie as well as the \textit{terminal} nodes, that correspond to non-empty suffixes of $y$, become {\em explicit} nodes of the suffix tree, while the other nodes are {\em implicit}. We let $\mathcal{L}(v)$ denote the \textit{path-label} from the root node to node $v$. We say that node $v$ is path-labeled $\mathcal{L}(v)$; i.e., the concatenation of the edge labels along the path from the root node to $v$. Additionally, $\mathcal{D}(v)= |\mathcal{L}(v)|$ is used to denote the \textit{word-depth} of node $v$. A node $v$ such that the path-label $\mathcal{L}(v) = y[i\dd n-1]$, for some $0 \leq i \leq n-1$, is {\em terminal} and is also labeled with index $i$. Each factor of $y$ is uniquely represented by either an explicit or an implicit node of $\mathcal{T}(y)$ called its \emph{locus}. The \textit{suffix-link} of a node $v$ with path-label $\mathcal{L}(v)= a w$ is a pointer to the node path-labeled $w$, where $a \in \Sigma$ is a single letter and $w$ is a word. The suffix-link of $v$ exists by construction if $v$ is a non-root branching node of $\mathcal{T}(y)$. 
The {\em matching statistics} of a word $x[0\dd |x|-1]$ with respect to word $y$ is an array $\text{MS}_x[0\dd |x|-1]$, where $\text{MS}_x[i]$ is a pair $(f_i, p_i)$ such that (i) $x[i  \dd i + f_i-1]$ is the longest prefix of $x[i  \dd |x|-1]$ that is a factor of $y$; and (ii) $y[p_i\dd p_i + f_i-1] = x[i\dd i + f_i-1]$~\cite{def}. $\mathcal{T}(y)$ is constructible in time $\cO(n)$, and, given $\mathcal{T}(y)$, we can compute $\text{MS}_x$ in time $\cO(|x|)$~\cite{def}.

\section{Combinatorial Properties}

For convenience, we consider the following setting. Let $y_1, y_2$ be words over the alphabet $\Sigma$ and let $y_3=y_1\#y_2$, with $\# \notin \Sigma$. Let $\ell$ be a positive integer and set $\Mina^{\ell}_{y_1} = \Mina_{y_1} \cap \Sigma^{\leq\ell}$ and $\Mina^{\ell}_{y_2} = \Mina_{y_2} \cap \Sigma^{\leq\ell}$. We want to construct $\Mina^{\ell}_{y_3} = \Mina_{y_3} \cap \Sigma^{\leq\ell}$. Let  $x\in\Mina^{\ell}_{y_3}$. We have two cases: 

\begin{description}
\item[Case 1]: $x \in \Mina^{\ell}_{y_1} \cup \Mina^{\ell}_{y_2}$;
\item[Case 2]: $x \notin \Mina^{\ell}_{y_1} \cup \Mina^{\ell}_{y_2}$.
\end{description}

The following auxiliary fact follows directly from the minimality property.

\begin{fact}\label{f1}
Word $x$ is absent from word $y$ if and only if $x$ is a superword of a $\MAW$ of $y$.
\end{fact}

For Case 1, we prove the following lemma.

\begin{lemma}[Case 1]\label{lem:C1} 
 A word  $x\in\Mina^{\ell}_{y_1}$ (resp.~$x\in \Mina^{\ell}_{y_2}$) belongs to $\Mina^{\ell}_{y_3}$ if and only if $x$ is a superword of a word in $\Mina^{\ell}_{y_2}$ (resp.~in $\Mina^{\ell}_{y_1}$). 
\end{lemma}

\begin{proof}
Let  $x\in\Mina^{\ell}_{y_1}$ (the case  $x\in\Mina^{\ell}_{y_2}$ is symmetric). Suppose first that $x$ is a superword of a word in $\Mina^{\ell}_{y_2}$, that is, there exists $v\in\Mina^{\ell}_{y_2}$ such that $v$ is a factor of $x$. If $v=x$, then $x\in \Mina^{\ell}_{y_1}\cap \Mina^{\ell}_{y_2}$ and therefore, using the definition of $\MAW$, $x\in \Mina^{\ell}_{y_3}$. If $v$ is a proper factor of $x$, then $x$ is an absent word of $y_2$ and again, by definition of $\MAW$, $x\in \Mina^{\ell}_{y_3}$.

Suppose now that $x$ is not a superword of any word in $\Mina^{\ell}_{y_2}$. Then $x$ is not absent in $y_2$ by Fact~\ref{f1}, and hence in $y_3$, thus $x$ cannot belong to $\Mina^{\ell}_{y_3}$.
\end{proof}

It should be clear that the statement of Lemma~\ref{lem:C1} implies, in particular, that all words in $\Mina^{\ell}_{y_1} \cap \Mina^{\ell}_{y_2}$ belong to $\Mina^{\ell}_{y_3}$. Furthermore, Lemma~\ref{lem:C1} motivates us to introduce the \emph{reduced set of $\MAWs$} of $y_1$ with respect to $y_2$ as the set $\RMina^{\ell}_{y_1}$ obtained from $\Mina^{\ell}_{y_1}$  after removing those words that are superwords of words in $\Mina^{\ell}_{y_2}$. The set $\RMina^{\ell}_{y_2}$ is defined analogously.





\begin{example}
Let $y_1=\texttt{abaab}$, $y_2=\texttt{bbaaab}$ and $\ell=5$.
We have $\Mina^{\ell}_{y_1}=\{\texttt{bb,aaa,bab,aaba}\}$ and $\Mina^{\ell}_{y_2}=\{\texttt{bbb,aaaa,baab,aba,bab,abb}\}.$
The word $\texttt{bab}$ is contained in  $\Mina^{\ell}_{y_1}\cap \Mina^{\ell}_{y_2}$ so it belongs to $\Mina^{\ell}_{y_3}$. The word $\texttt{aaba}\in \Mina^{\ell}_{y_1}$ is a superword of $\texttt{aba}\in \Mina^{\ell}_{y_2}$ hence $\texttt{aaba}\in \Mina^{\ell}_{y_3}$. On the other hand, the words $\texttt{bbb}$, $\texttt{aaaa}$ and $\texttt{abb}$ are superwords of words in $\Mina^{\ell}_{y_1}$, hence they belong to $\Mina^{\ell}_{y_3}$. The remaining \MAWs{} are not superwords of \MAWs{} of the other word. The reduced sets are therefore $\RMina^{\ell}_{y_1}=\{\texttt{bb},\texttt{aaa}\}$ and $\RMina^{\ell}_{y_2}=\{\texttt{baab},\texttt{aba}\}$.
In conclusion, we have for Case 1 that
$\Mina^{\ell}_{y_3}\cap (\Mina^{\ell}_{y_1}\cup \Mina^{\ell}_{y_2} )=\{\texttt{aaaa,bab,aaba,abb,bbb}\}$.\qed
\end{example}


We now investigate the set $ \Mina^{\ell}_{y_3}\setminus ( \Mina^{\ell}_{y_1}\cup  \Mina^{\ell}_{y_2})$ (Case 2).

\begin{fact}\label{fct:factors} Let $x = aub$, $a,b\in \Sigma$, be such that $x \in \Mina^{\ell}_{y_3}$ and $x \notin \Mina^{\ell}_{y_1} \cup \Mina^{\ell}_{y_2}$. Then $au$ occurs  in $y_1$ but not in $y_2$ and $ub$ occurs  in $y_2$ but not in $y_1$, or vice versa. 
\end{fact}



The rationale for generating the reduced sets should become clear with the next lemma.




\begin{lemma}[Case 2]\label{lem:cartesian}
Let $x \in \Mina^{\ell}_{y_3} \setminus (\Mina^{\ell}_{y_1}\cup  \Mina^{\ell}_{y_2})$. Then $x$ has a prefix $x_i$ in $\RMina^{\ell}_{y_i}$ and a suffix $x_j$ in $ \RMina^{\ell}_{y_j}$, for $i,j$ such that $\{i,j\}=\{1,2\}$.
\end{lemma}

\begin{proof}
Let $x=aub$, $a,b\in\Sigma$, be a word in $\Mina^{\ell}_{y_3} \setminus (\Mina^{\ell}_{y_1}\cup  \Mina^{\ell}_{y_2})$. By Fact~\ref{fct:factors}, $au$ occurs  in $y_1$ but not in $y_2$ and $ub$ occurs  in $y_2$ but not in $y_1$, or vice versa. Let us assume the first case holds (the other case is symmetric). Since $au$ does not occur in $y_2$, there is a \MAW{} $x_2 \in \Mina^{\ell}_{y_2}$ that is a factor of $au$. Since $ub$ occurs in $y_2$, $x_2$ is not a factor of $ub$. Consequently, $x_2$ is a prefix of $au$.

Analogously, there is an $x_1 \in \Mina^{\ell}_{y_1}$ that is a suffix of $ub$. Furthermore, $x_1$ and $x_2$ cannot be factors one of another. Inspect Figure~\ref{fig:remark} in this regard.
\end{proof}

\begin{figure}[!t]
\begin{center}
\includegraphics[width=9cm]{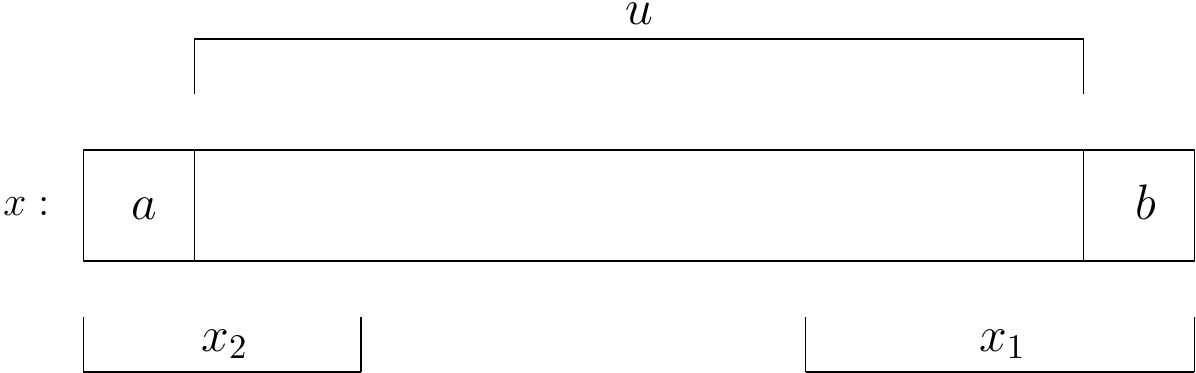}
\end{center}
\caption{$x_2$ occurs in $y_1$ but not in $y_2$; $x_1$ occurs in $y_2$ but not in $y_1$; therefore $aub$ does not occur in $y_1\#y_2$. By construction, $au$ occurs in $y_1$ and $ub$ occurs in $y_2$; therefore $aub$ is a Case 2 \MAW.}
\label{fig:remark}
\end{figure}
 
\begin{example}
Let $y_1=\texttt{abaab}$, $y_2=\texttt{bbaaab}$ and $\ell=5$. Consider $x=\texttt{abaaa}\in \Mina^{\ell}_{y_3} \setminus (\Mina^{\ell}_{y_1}\cup  \Mina^{\ell}_{y_2})$ (Case 2 \MAW).
We have that $\texttt{abaa}$ occurs in $y_1$ but not in $y_2$ and $\texttt{baaa}$ occurs in $y_2$ but not in $y_1$. Since $\texttt{abaa}$ does not occur in $y_2$, there is a \MAW{} $x_2 \in \RMina^{\ell}_{y_2}$ that is a factor of $\texttt{abaa}$. Since $\texttt{baaa}$ occurs in $y_2$, $x_2$ is not a factor of $\texttt{baaa}$. So $x_2$ is  a prefix of $\texttt{abaa}$ and this is $\texttt{aba}$.
Analogously, there is \MAW{} $x_1 \in \RMina^{\ell}_{y_1}$ that is a suffix of $\texttt{abaaa}$ and this is $\texttt{aaa}$. \qed
\end{example}

As a consequence of Lemma~\ref{lem:cartesian}, in order to construct the set $\Mina^{\ell}_{y_3} \setminus (\Mina^{\ell}_{y_1}\cup  \Mina^{\ell}_{y_2})$, we should consider all pairs $(x_i,x_j)$ with $x_i$ in $ \RMina^{\ell}_{y_i}$ and $x_j$ in $ \RMina^{\ell}_{y_j}$, $\{i,j\}=\{1,2\}$. In order to construct the final set $\Mina^{\ell}_{y_1\#\ldots\#y_N}$, we use incrementally Lemmas~\ref{lem:C1} and \ref{lem:cartesian}. We summarise the whole approach in the following general theorem, which forms the theoretical basis of our technique.

\begin{theorem}\label{thm:general}
Let $N>1$, and let $x\in \Mina^{\ell}_{y_1\#\ldots\#y_N}$. Then, either $x\in \Mina^{\ell}_{y_1\#\ldots\#y_{N-1}}\cup \Mina^{\ell}_{y_N}$ (Case 1 \MAWs) or, otherwise, $x\in \Mina^{\ell}_{y_i\#y_N}\setminus (\Mina^{\ell}_{y_i}\cup \Mina^{\ell}_{y_N})$ for some $i$. Moreover, in this latter case, $x$ has a prefix in $\RMina^{\ell}_{y_1\#\ldots\#y_{N-1}}$ and a suffix in $\RMina^{\ell}_{y_N}$, or the converse, i.e., $x$ has a prefix in $\RMina^{\ell}_{y_N}$ and a suffix in $\RMina^{\ell}_{y_1\#\ldots\#y_{N-1}}$ (Case 2 \MAWs).
 \end{theorem}
\begin{proof}
Let $x\in \Mina^{\ell}_{y_1\#\ldots\#y_N}$ and $x\notin \Mina^{\ell}_{y_1\#\ldots\#y_{N-1}}\cup \Mina^{\ell}_{y_N}$. 
Then,
 $x\notin \Mina^{\ell}_{y_1\#\ldots\#y_{N-1}}$ and 
 $x \notin \Mina^{\ell}_{y_N}$.

\noindent
Let $x$ be a word of length $m$. By the definition of $\MAW$, $x[0\dd m-2]$ and $x[1\dd m-1]$ must both be factors of $y_1\#\ldots\#y_N$. However, both cannot be factors of $y_1\#\ldots\#y_{N-1}$ and both cannot be factors of $y_N$. Therefore, we have one of the two cases:
\begin{description}
\item[Case 1]: $x[0\dd m-2]$ is factor of $y_1\#\ldots\#y_{N-1}$ but not of $y_N$ and $x[1\dd m-1]$ is a factor of $y_N$ but not of $y_1\#\ldots\#y_{N-1}$.
\item[Case 2]: $x[0\dd m-2]$ is factor of $y_N$ but not of $y_1\#\ldots\#y_{N-1}$ and $x[1\dd m-1]$ is a factor of $y_1\#\ldots\#y_{N-1}$ but not of $y_N$.
\end{description}

These two cases are symmetric, thus only proof of Case 1 will be presented here.
If $x[0]$ does not occur in $y_N$ then  $x[0] \in \RMina^{\ell}_{y_N}$. Otherwise, let $x[0\dd t]$ be the longest prefix of $x[0\dd m-2]$ that is a factor of $y_N$. 

Because  $0 \leq t<m-1$ then $x[1\dd t+1]$ is a factor of $y_N$. Therefore, $x[0\dd t+1] \in \Mina^{\ell}_{y_N}$. In addition, all factors of $x[0\dd t+1]$ occur in $y_1\#\ldots\#y_{N-1}$, so $x[0\dd t+1] \in \RMina^{\ell}_{y_N}$.

Now, $x[1\dd m-1]$ does not occur in $y_1\#\ldots\#y_{N-1}$, so either $x[m-1]$ does not occur in $y_1\#\ldots\#y_{N-1}$ which means that $x[m-1] \in \RMina^{\ell}_{y_1\#\ldots\#y_{N-1}}$, or let $x[p\dd m-1]$ be the longest suffix of $x[1\dd m-1]$ that occurs in $y_1\#\ldots\#y_{N-1}$. 

Because $0<p \leq m-1$ then $x[p-1\dd m-2]$ occurs in $y_1\#\ldots\#y_{N-1}$, therefore $x[p-1\dd m-1] \in \Mina^{\ell}_{y_1\#\ldots\#y_{N-1}}$. Since all factors of $x[p-1\dd m-1]$ occur in $y_N$, we have $x[p-1\dd m-1] \in \RMina^{\ell}_{y_1\#\ldots\#y_{N-1}}$.
\end{proof}

\section{Algorithm}\label{sec:time-efficient}

Let us first introduce an algorithmic tool. In the \textit{weighted ancestor} problem, introduced in~\cite{DBLP:conf/cpm/FarachM96}, we consider a rooted tree $T$ with an integer weight function $\mu$ defined on the nodes. We require that the weight of the root is zero and the weight of any other node is strictly larger than the weight of its parent. A weighted ancestor query, given a node $v$ and an integer value $w\le \mu(v)$, asks for the highest ancestor $u$ of $v$ such that $\mu(u)\ge w$, i.e., such an ancestor $u$ that $\mu(u)\ge w$ and $\mu(u)$ is the smallest possible. When $T$ is the suffix tree of a word $y$ of length $n$, we can locate the locus of any factor $y[i\dd j]$ using a weighted ancestor query. We define the weight of a node of the suffix tree as the length of the word it represents. Thus a weighted ancestor query can be used for the terminal node decorated with $i$ to create (if necessary) and mark the node that corresponds to $y[i\dd j]$. 

\begin{theorem}[\cite{DBLP:journals/corr/BartonKLPR17}]\label{the:WLA}
Given a collection $Q$ of weighted ancestor queries on a weighted tree $T$ on $n$ nodes with integer weights up to $n^{\cO(1)}$, all the queries in $Q$ can be answered {\em off-line} in $\cO(n+|Q|)$ time.
\end{theorem}

\subsection{The Algorithm}

At the $N$th step, we have in memory the set $\Mina^{\ell}_{y_1\#\ldots\#y_{N-1}}$.
Our algorithm works as follows:

\begin{enumerate}
\item We read word $y_N$ from the disk and compute $\Mina^{\ell}_{y_N}$ in time $\cO(|y_N|)$.
We output the words in the following constant-space form: $<i_1,i_2,\alpha>$ per word~\cite{DBLP:journals/bmcbi/BartonHMP14}; such that $y_N[i_1\dd i_2]\cdot \alpha\in \Mina^{\ell}_{y_N}$.

\item Here we compute Case 1 \MAWs. We apply Lemma~\ref{lem:C1} to construct set $M=\{w: w \in \Mina^{\ell}_{y_{1}\#\ldots\#y_{N}}, w\in \Mina^{\ell}_{y_{1}\#\ldots\#y_{N-1}} \cup \Mina^{\ell}_{y_N}\}$ and the sets $\RMina^{\ell}_{y_{1}\#\ldots\#y_{N-1}}, \RMina^{\ell}_{y_N}$ as follows. 

\begin{enumerate}
\item We first want to find the elements of $\Mina^{\ell}_{y_{1}\#\ldots\#y_{N-1}}$ that are superwords of any word $y_N[i_1\dd i_2]\cdot \alpha$. We build the generalised suffix tree $T_1=\mathcal{T}(\Mina^{\ell}_{y_{1}\#\ldots\#y_{N-1}}\cup\{y_N\})$~\cite{def}.
We find the locus of the longest proper prefix $y_N[i_1\dd i_2]$ of each element of $\Mina^{\ell}_{y_N}$ in $T_1$ via answering off-line a batch of weighted ancestor queries (Theorem~\ref{the:WLA}). From there on, we spell $\alpha$ and mark the corresponding node on $T_1$, if any. 
After processing all $<i_1,i_2,\alpha>$ in the same manner, we traverse $T_1$ to gather all occurrences (starting positions) of words $y_N[i_1\dd i_2]\cdot \alpha$ in the elements of $\Mina^{\ell}_{y_{1}\#\ldots\#y_{N-1}}$, thus finding the elements of $\Mina^{\ell}_{y_{1}\#\ldots\#y_{N-1}}$ that are superwords of any $y_N[i_1\dd i_2]\cdot \alpha$. By definition, no \MAW{} $y_N[i_1\dd i_2]\cdot \alpha$ is a prefix of another \MAW{} $y_N[i_1'\dd i_2']\cdot \alpha'$, thus the marked nodes form pairwise disjoint subtrees, and the whole process takes time $\cO(|y_N| + ||\Mina^{\ell}_{y_{1}\#\ldots\#y_{N-1}}||)$, the size of $T_1$.

\item We next want to check if the words $y_N[i_1\dd i_2]\cdot \alpha$ are superwords of any element of $\Mina^{\ell}_{y_{1}\#\ldots\#y_{N-1}}$. We first sort all tuples $<i_1,i_2,\alpha>$ using radixsort and then check this using the matching statistics algorithm for $y_N$ with respect to $\mathcal{T}(\Mina^{\ell}_{y_{1}\#\ldots\#y_{N-1}})$ considering the tuples in ascending order (from left to right) at the same time. By definition, no element in $\Mina^{\ell}_{y_{1}\#\ldots\#y_{N-1}}$ is a factor of another element in the same set. Thus if a factor of $y_N[i_1\dd i_2]\cdot \alpha$ corresponds to an element in $\Mina^{\ell}_{y_{1}\#\ldots\#y_{N-1}}$ this is easily located in $\mathcal{T}(\Mina^{\ell}_{y_{1}\#\ldots\#y_{N-1}})$ while running the matching statistics algorithm. The whole process takes $\cO(|y_N| + ||\Mina^{\ell}_{y_{1}\#\ldots\#y_{N-1}}||)$ time: $\cO(||\Mina^{\ell}_{y_{1}\#\ldots\#y_{N-1}}||)$ time to construct the suffix tree and a further $\cO(|y_N|)$ time for the matching statistics algorithm and for processing the $\cO(|y_N|)$ tuples.

\end{enumerate}
We create set $\RMina^{\ell}_{y_{1}\#\ldots\#y_{N-1}}$ explicitly since it is a subset of $\Mina^{\ell}_{y_{1}\#\ldots\#y_{N-1}}$. We create set $\RMina^{\ell}_{y_{N}}$ implicitly: every element $x \in \RMina^{\ell}_{y_N}$ is stored as a tuple $<i_1,i_2,\alpha>$ such that $x=y_N[i_1\dd i_2]\cdot\alpha$. We store every element of $\{x_2: x_2 \in M \cap \Mina^{\ell}_{y_{N}}\}$ with the same representation. All other elements of $M$ are stored explicitly.

\item Construct the suffix tree of $y_N$ and use it to locate all occurrences of words in $\RMina^{\ell}_{y_1\# \ldots \#y_{N-1}}$ in $y_N$ and store the occurrences as pairs (starting position, ending position). This step can be done in time $\cO(|y_N| + ||\RMina^{\ell}_{y_{1}\#\ldots\#y_{N-1}}||)$.
By definition, no element in $\RMina^{\ell}_{y_1\# \ldots \#y_{N-1}}$ is a prefix of another element in $\RMina^{\ell}_{y_1\# \ldots \#y_{N-1}}$, and thus this can be done within the claimed time complexity.

\item For every $i\in[1,N-1]$, we perform the following to compute Case 2 \MAWs:
\begin{enumerate}
\item Read word $y_i$ from the disk. Construct the suffix tree $T_{x}$ of word $x=y_i\#y_N$ in time $\cO(|y_i| + |y_N|)$. Use $T_{x}$ to locate all occurrences of elements of $\RMina^{\ell}_{y_N}$ in $y_i$ and store the occurrences as pairs (starting position, ending position). This step can be done in time $\cO(|y_i| + |y_N|)$ similar to step 2. By definition, no element in $\RMina^{\ell}_{y_N}$ is a prefix of another element in $\RMina^{\ell}_{y_N}$, and thus this can be done within the claimed time complexity.

\item During a bottom-up traversal of $T_{x}$ mark, at each explicit node of $T_{x}$, the smallest starting position of the subword represented by that node, and the largest starting position of the same subword. This can be done in time $\cO(|y_i| + |y_N|)$ by propagating upwards the labels of the terminal nodes (starting positions of suffixes) and updating the smallest and largest positions analogously.

\item Compute the set $\Mina^{\ell}_{y_i\#y_N}$ and output the words in the following constant-space form: $<a,i_1,i_2,b>$ per word; such that $a\cdot x[i_1\dd i_2]\cdot b$ is a $\MAW$. This can be done in time $\cO(|y_i| + |y_N|)$.
\item For each element of $\Mina^{\ell}_{y_i\#y_N}$, we need to locate the node representing word $ax[i_1\dd i_2]=au$ and the node representing word $x[i_1\dd i_2]b=ub$. This can be done in time $\cO(|y_i| + |y_N|)$ via answering off-line a batch of weighted ancestor queries (Theorem~\ref{the:WLA}). 
At this point, we have located the two nodes on $T_x$. We assign a pointer from the stored starting position $g$ of $au$ to the ending position $f$ of $ub$, only if $g$ is before $\#$ and $f$ is after $\#$ ($f$ can be trivially computed using the stored starting position of $ub$ and the length of $ub$). Conversely, we assign a pointer from the ending position $f$ of $ub$ to the stored starting position $g$ of $au$, only if $f$ is before $\#$ and $g$ is after $\#$. 
\item Suppose $au$ occurs in $y_i$ and $ub$ in $y_N$. We make use of the pointers as follows.
Recall steps 3 and 4(a) and check whether $au$ starts where a word $r_1$ of $\RMina^{\ell}_{y_N}$ starts and $ub$ ends where a word $r_2$ of $\RMina^{\ell}_{y_1\#\ldots\#y_{N-1}}$ ends. If this is the case and $|u| \geq \max\{|r_1|,|r_2|\}-1$, then by Theorem~\ref{thm:general} $aub$ is added to our output set $M$, otherwise discard it. Inspect Figure~\ref{fig:case2} in this regard. Conversely, if $au$ occurs in $y_N$ and $ub$ in $y_i$ check whether $au$ starts where a word $r_2$ of $\RMina^{\ell}_{y_1\#\ldots\#y_{N-1}}$ starts and whether $ub$ ends where a word $r_1$ of $\RMina^{\ell}_{y_N}$ ends. If this is the case and $|u| \geq \max\{|r_1|,|r_2|\}-1$, then $aub$ is added to $M$, otherwise discard it. 
\end{enumerate}
\end{enumerate}

\begin{figure}[!t]
\begin{center}
\includegraphics[width=10cm]{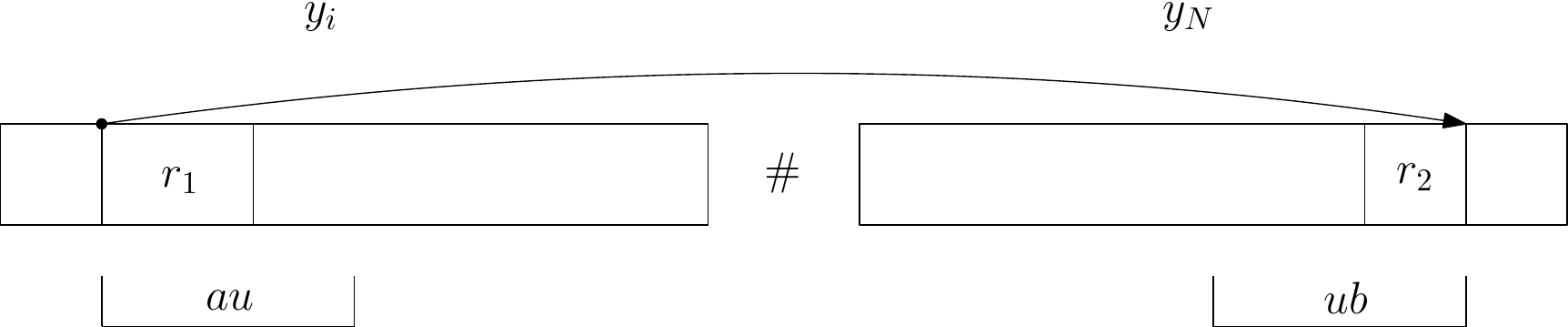}
\end{center}
\caption{$au$ starts where a word $r_1$ of $\RMina^{\ell}_{y_N}$ starts in $y_i$ and $ub$ ends where a word $r_2$ of $\RMina^{\ell}_{y_1\#\ldots\#y_{N-1}}$ ends in $y_N$. Moreover, if $|u| \geq \max\{|r_1|,|r_2|\}-1$, then $aub$ is a Case 2 \MAW.}
\label{fig:case2}
\end{figure}

Finally, we set $\Mina^{\ell}_{y_{1}\#\ldots\#y_{N}}=M$ as the output of the $N$th step. Let $\textsc{MaxIn}$ be the length of the longest word in $\{y_1,\ldots,y_k\}$ and $\textsc{MaxOut}=\max\{||\Mina^{\ell}_{y_{1}\#\ldots\#y_{N}}||:N\in[1,k]\}$. 

\begin{theorem}\label{the:main_result}
Given $k$ words $y_1,y_2,\ldots,y_k$ and an integer $\ell>0$, all $\Mina^{\ell}_{y_{1}},\ldots,\Mina^{\ell}_{y_{1}\#\ldots\#y_{k}}$ can be computed in $\cO(kn+\sum^{k}_{N=1}||\Mina^{\ell}_{y_{1}\#\ldots\#y_{N}}||)$ total time using $\cO(\textsc{MaxIn}+\textsc{MaxOut})$ space, where $n=|y_{1}\#\ldots\#y_{k}|$. 
\end{theorem}

\begin{proof}
From the above discussion, the time is bounded by $\cO(\sum^{k}_{N=1}\sum^{N-1}_{i=1}(|y_N|+|y_i|)+\sum^{k}_{N=1}||\Mina^{\ell}_{y_{1}\#\ldots\#y_{N}}||)$. We can bound the first term as follows.
$$\sum^{k}_{N=1}\sum^{N-1}_{i=1}(|y_N|+|y_i|)\leq\sum^{k}_{N=1}\sum^{k}_{i=1}(|y_N|+|y_i|)=\sum^{k}_{N=1}\sum^{k}_{i=1}|y_N|+\sum^{k}_{N=1}\sum^{k}_{i=1}|y_i|=2k(|y_1|+\cdots+|y_k|).$$
Therefore the time is bounded by $\cO(kn+\sum^{k}_{N=1}||\Mina^{\ell}_{y_{1}\#\ldots\#y_{N}}||)$.

The space is bounded by the maximum time spent at a single step; namely, the length of the longest word in the collection plus the maximum total size of set elements across all output sets. 
Note that the total output size of the algorithm is the sum of all its output sets, that is $\sum^{k}_{N=1}||\Mina^{\ell}_{y_{1}\#\ldots\#y_{N}}||$, and \textsc{MaxOut} could come from any intermediate set.

The correctness of the algorithm follows from Lemma~\ref{lem:C1} and Theorem~\ref{thm:general}.
\end{proof}


\section{Proof-of-Concept Experiments}

In this section, we do not directly compare against the fastest internal~\cite{DBLP:journals/bmcbi/BartonHMP14} or external~\cite{DBLP:journals/bioinformatics/HeliouPP17} memory implementations because the former assumes that we have the required amount of internal memory, and the latter assumes that we have the required amount of external memory to construct and store the global data structures for a given input dataset. If the memory for constructing and storing the data structures is available, these linear-time algorithms are surely faster than the method proposed here. In what follows, we rather show that our output-sensitive technique offers a space-time tradeoff, which can be usefully exploited for specific values of $\ell$, the maximal length of $\MAW$s we wish to compute.

The algorithm discussed in Section~\ref{sec:time-efficient} (with the exception of storing and searching the reduced set words explicitly rather than in the constant-space form previously described) has been implemented in the \texttt{C++} programming language\footnote{The implementation can be made available upon request.}. The correctness of our implementation has been confirmed against that of~\cite{DBLP:journals/bmcbi/BartonHMP14}.
As input dataset here we used the entire human genome (version hg38)~\cite{humangenome}, which has an approximate size of 3.1GB. The following experiment was conducted on a machine with an Intel Core i5-4690 CPU at 3.50 GHz and 128GB of memory running GNU/Linux. We ran the program by splitting the genome into $k={2,4,6,8,10}$ blocks and setting $\ell = 10,11,12$. Figure~\ref{exp} depicts the change in elapsed time and peak memory usage as $k$ and $\ell$ increase (space-time tradeoff). 

Graph (a) shows an increase of time as $k$ and $\ell$ increase; and graph (b) shows a decrease in memory as $k$ increases (as proved in Theorem~\ref{the:main_result}). Notice that the space to construct the block-wise data structures bounds the total space used for the specific $\ell$ values and that is why the memory peak is essentially the same for the $\ell$ values used. This can specifically be seen for $\ell=10$ where all words of length $10$ are present in the genome. The same dataset was used to run the fastest internal memory implementation for computing $\MAW$s~\cite{DBLP:journals/bmcbi/BartonHMP14} on the same machine. It took only $2242$ seconds to compute all $\MAW$s but with a peak memory usage of $60.80$GB. The results confirm our theoretical findings and justify our contribution.

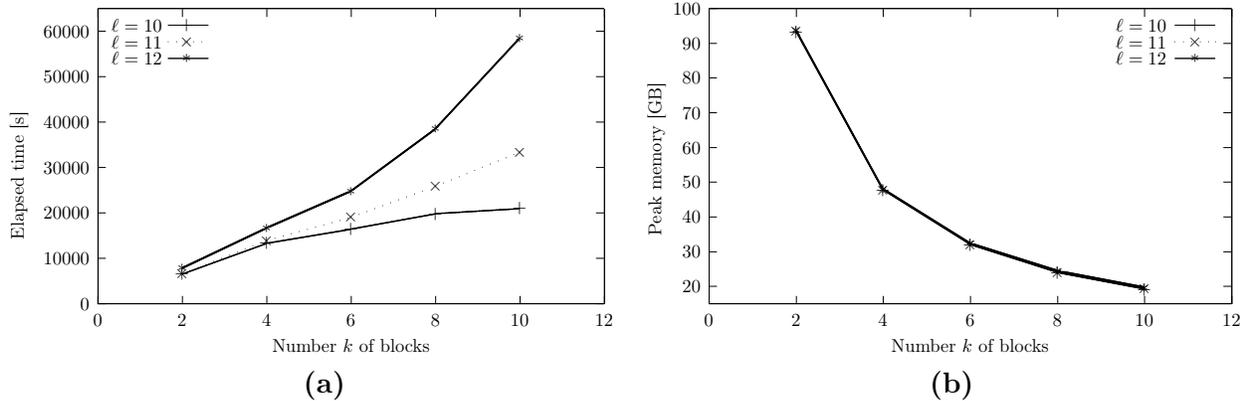
\begin{figure}[!t] 
  \begin{minipage}[b]{0.5\textwidth}
    \centering
   \resizebox{1\textwidth}{!}{\input{time_acc.tex}}
    \textbf{(a)}
  \end{minipage}
  \begin{minipage}[b]{0.5\textwidth}
    \centering
     \resizebox{1\textwidth}{!}{\input{memory_acc.tex}}
     \textbf{(b)}
  \end{minipage}
  \caption{Elapsed time and peak memory usage using increasing $k$ blocks of the entire human genome for $\ell=10,11,12$; notice that the peak memory usage is the same for all values of $\ell$.}
  \label{exp}
\end{figure}

\bibliographystyle{plain}
\bibliography{refs}

\end{document}

%% file: time_acc.tex
\setlength{\unitlength}{0.240900pt}
\ifx\plotpoint\undefined\newsavebox{\plotpoint}\fi
\sbox{\plotpoint}{\rule[-0.200pt]{0.400pt}{0.400pt}}%
\begin{picture}(1500,900)(0,0)
\sbox{\plotpoint}{\rule[-0.200pt]{0.400pt}{0.400pt}}%
\put(191.0,131.0){\rule[-0.200pt]{4.818pt}{0.400pt}}
\put(171,131){\makebox(0,0)[r]{$0$}}
\put(1419.0,131.0){\rule[-0.200pt]{4.818pt}{0.400pt}}
\put(191.0,243.0){\rule[-0.200pt]{4.818pt}{0.400pt}}
\put(171,243){\makebox(0,0)[r]{$10000$}}
\put(1419.0,243.0){\rule[-0.200pt]{4.818pt}{0.400pt}}
\put(191.0,355.0){\rule[-0.200pt]{4.818pt}{0.400pt}}
\put(171,355){\makebox(0,0)[r]{$20000$}}
\put(1419.0,355.0){\rule[-0.200pt]{4.818pt}{0.400pt}}
\put(191.0,467.0){\rule[-0.200pt]{4.818pt}{0.400pt}}
\put(171,467){\makebox(0,0)[r]{$30000$}}
\put(1419.0,467.0){\rule[-0.200pt]{4.818pt}{0.400pt}}
\put(191.0,579.0){\rule[-0.200pt]{4.818pt}{0.400pt}}
\put(171,579){\makebox(0,0)[r]{$40000$}}
\put(1419.0,579.0){\rule[-0.200pt]{4.818pt}{0.400pt}}
\put(191.0,691.0){\rule[-0.200pt]{4.818pt}{0.400pt}}
\put(171,691){\makebox(0,0)[r]{$50000$}}
\put(1419.0,691.0){\rule[-0.200pt]{4.818pt}{0.400pt}}
\put(191.0,803.0){\rule[-0.200pt]{4.818pt}{0.400pt}}
\put(171,803){\makebox(0,0)[r]{$60000$}}
\put(1419.0,803.0){\rule[-0.200pt]{4.818pt}{0.400pt}}
\put(191.0,131.0){\rule[-0.200pt]{0.400pt}{4.818pt}}
\put(191,90){\makebox(0,0){$0$}}
\put(191.0,839.0){\rule[-0.200pt]{0.400pt}{4.818pt}}
\put(399.0,131.0){\rule[-0.200pt]{0.400pt}{4.818pt}}
\put(399,90){\makebox(0,0){$2$}}
\put(399.0,839.0){\rule[-0.200pt]{0.400pt}{4.818pt}}
\put(607.0,131.0){\rule[-0.200pt]{0.400pt}{4.818pt}}
\put(607,90){\makebox(0,0){$4$}}
\put(607.0,839.0){\rule[-0.200pt]{0.400pt}{4.818pt}}
\put(815.0,131.0){\rule[-0.200pt]{0.400pt}{4.818pt}}
\put(815,90){\makebox(0,0){$6$}}
\put(815.0,839.0){\rule[-0.200pt]{0.400pt}{4.818pt}}
\put(1023.0,131.0){\rule[-0.200pt]{0.400pt}{4.818pt}}
\put(1023,90){\makebox(0,0){$8$}}
\put(1023.0,839.0){\rule[-0.200pt]{0.400pt}{4.818pt}}
\put(1231.0,131.0){\rule[-0.200pt]{0.400pt}{4.818pt}}
\put(1231,90){\makebox(0,0){$10$}}
\put(1231.0,839.0){\rule[-0.200pt]{0.400pt}{4.818pt}}
\put(1439.0,131.0){\rule[-0.200pt]{0.400pt}{4.818pt}}
\put(1439,90){\makebox(0,0){$12$}}
\put(1439.0,839.0){\rule[-0.200pt]{0.400pt}{4.818pt}}
\put(191.0,131.0){\rule[-0.200pt]{0.400pt}{175.375pt}}
\put(191.0,131.0){\rule[-0.200pt]{300.643pt}{0.400pt}}
\put(1439.0,131.0){\rule[-0.200pt]{0.400pt}{175.375pt}}
\put(191.0,859.0){\rule[-0.200pt]{300.643pt}{0.400pt}}
\put(-25,290){\rotatebox{90}{Elapsed time [s]}}
\put(815,29){\makebox(0,0){Number $k$ of blocks}}
\put(350,818){\makebox(0,0)[r]{$\ell=10$}}
\put(364.0,818.0){\rule[-0.200pt]{24.090pt}{0.400pt}}
\put(399,204){\usebox{\plotpoint}}
\multiput(399.00,204.58)(1.372,0.499){149}{\rule{1.195pt}{0.120pt}}
\multiput(399.00,203.17)(205.520,76.000){2}{\rule{0.597pt}{0.400pt}}
\multiput(607.00,280.58)(2.996,0.498){67}{\rule{2.477pt}{0.120pt}}
\multiput(607.00,279.17)(202.859,35.000){2}{\rule{1.239pt}{0.400pt}}
\multiput(815.00,315.58)(2.757,0.498){73}{\rule{2.289pt}{0.120pt}}
\multiput(815.00,314.17)(203.248,38.000){2}{\rule{1.145pt}{0.400pt}}
\multiput(1023.00,353.58)(8.228,0.493){23}{\rule{6.500pt}{0.119pt}}
\multiput(1023.00,352.17)(194.509,13.000){2}{\rule{3.250pt}{0.400pt}}
\put(399,204){\makebox(0,0){$+$}}
\put(607,280){\makebox(0,0){$+$}}
\put(815,315){\makebox(0,0){$+$}}
\put(1023,353){\makebox(0,0){$+$}}
\put(1231,366){\makebox(0,0){$+$}}
\put(410,818){\makebox(0,0){$+$}}
\put(350,777){\makebox(0,0)[r]{$\ell=11$}}
\multiput(364,777)(20.756,0.000){5}{\usebox{\plotpoint}}
\put(411,777){\usebox{\plotpoint}}
\put(399,205){\usebox{\plotpoint}}
\multiput(399,205)(19.309,7.612){11}{\usebox{\plotpoint}}
\multiput(607,287)(20.017,5.486){11}{\usebox{\plotpoint}}
\multiput(815,344)(19.465,7.206){10}{\usebox{\plotpoint}}
\multiput(1023,421)(19.245,7.772){11}{\usebox{\plotpoint}}
\put(1231,505){\usebox{\plotpoint}}
\put(399,205){\makebox(0,0){$\times$}}
\put(607,287){\makebox(0,0){$\times$}}
\put(815,344){\makebox(0,0){$\times$}}
\put(1023,421){\makebox(0,0){$\times$}}
\put(1231,505){\makebox(0,0){$\times$}}
\put(410,777){\makebox(0,0){$\times$}}
\sbox{\plotpoint}{\rule[-0.400pt]{0.800pt}{0.800pt}}%
\sbox{\plotpoint}{\rule[-0.200pt]{0.400pt}{0.400pt}}%
\put(350,736){\makebox(0,0)[r]{$\ell=12$}}
\sbox{\plotpoint}{\rule[-0.400pt]{0.800pt}{0.800pt}}%
\put(364.0,736.0){\rule[-0.400pt]{24.090pt}{0.800pt}}
\put(399,219){\usebox{\plotpoint}}
\multiput(399.00,220.41)(1.054,0.501){191}{\rule{1.881pt}{0.121pt}}
\multiput(399.00,217.34)(204.096,99.000){2}{\rule{0.940pt}{0.800pt}}
\multiput(607.00,319.41)(1.147,0.501){175}{\rule{2.029pt}{0.121pt}}
\multiput(607.00,316.34)(203.790,91.000){2}{\rule{1.014pt}{0.800pt}}
\multiput(815.00,410.41)(0.680,0.501){299}{\rule{1.288pt}{0.121pt}}
\multiput(815.00,407.34)(205.328,153.000){2}{\rule{0.644pt}{0.800pt}}
\multiput(1024.41,562.00)(0.500,0.538){409}{\rule{0.121pt}{1.062pt}}
\multiput(1021.34,562.00)(208.000,221.797){2}{\rule{0.800pt}{0.531pt}}
\put(399,219){\makebox(0,0){$\ast$}}
\put(607,318){\makebox(0,0){$\ast$}}
\put(815,409){\makebox(0,0){$\ast$}}
\put(1023,562){\makebox(0,0){$\ast$}}
\put(1231,786){\makebox(0,0){$\ast$}}
\put(410,736){\makebox(0,0){$\ast$}}
\sbox{\plotpoint}{\rule[-0.200pt]{0.400pt}{0.400pt}}%
\put(191.0,131.0){\rule[-0.200pt]{0.400pt}{175.375pt}}
\put(191.0,131.0){\rule[-0.200pt]{300.643pt}{0.400pt}}
\put(1439.0,131.0){\rule[-0.200pt]{0.400pt}{175.375pt}}
\put(191.0,859.0){\rule[-0.200pt]{300.643pt}{0.400pt}}
\end{picture}

%% file: memory_acc.tex
\setlength{\unitlength}{0.240900pt}
\ifx\plotpoint\undefined\newsavebox{\plotpoint}\fi
\sbox{\plotpoint}{\rule[-0.200pt]{0.400pt}{0.400pt}}%
\begin{picture}(1500,900)(0,0)
\sbox{\plotpoint}{\rule[-0.200pt]{0.400pt}{0.400pt}}%
\put(151.0,174.0){\rule[-0.200pt]{4.818pt}{0.400pt}}
\put(131,174){\makebox(0,0)[r]{$20$}}
\put(1419.0,174.0){\rule[-0.200pt]{4.818pt}{0.400pt}}
\put(151.0,259.0){\rule[-0.200pt]{4.818pt}{0.400pt}}
\put(131,259){\makebox(0,0)[r]{$30$}}
\put(1419.0,259.0){\rule[-0.200pt]{4.818pt}{0.400pt}}
\put(151.0,345.0){\rule[-0.200pt]{4.818pt}{0.400pt}}
\put(131,345){\makebox(0,0)[r]{$40$}}
\put(1419.0,345.0){\rule[-0.200pt]{4.818pt}{0.400pt}}
\put(151.0,431.0){\rule[-0.200pt]{4.818pt}{0.400pt}}
\put(131,431){\makebox(0,0)[r]{$50$}}
\put(1419.0,431.0){\rule[-0.200pt]{4.818pt}{0.400pt}}
\put(151.0,516.0){\rule[-0.200pt]{4.818pt}{0.400pt}}
\put(131,516){\makebox(0,0)[r]{$60$}}
\put(1419.0,516.0){\rule[-0.200pt]{4.818pt}{0.400pt}}
\put(151.0,602.0){\rule[-0.200pt]{4.818pt}{0.400pt}}
\put(131,602){\makebox(0,0)[r]{$70$}}
\put(1419.0,602.0){\rule[-0.200pt]{4.818pt}{0.400pt}}
\put(151.0,688.0){\rule[-0.200pt]{4.818pt}{0.400pt}}
\put(131,688){\makebox(0,0)[r]{$80$}}
\put(1419.0,688.0){\rule[-0.200pt]{4.818pt}{0.400pt}}
\put(151.0,773.0){\rule[-0.200pt]{4.818pt}{0.400pt}}
\put(131,773){\makebox(0,0)[r]{$90$}}
\put(1419.0,773.0){\rule[-0.200pt]{4.818pt}{0.400pt}}
\put(151.0,859.0){\rule[-0.200pt]{4.818pt}{0.400pt}}
\put(131,859){\makebox(0,0)[r]{$100$}}
\put(1419.0,859.0){\rule[-0.200pt]{4.818pt}{0.400pt}}
\put(151.0,131.0){\rule[-0.200pt]{0.400pt}{4.818pt}}
\put(151,90){\makebox(0,0){$0$}}
\put(151.0,839.0){\rule[-0.200pt]{0.400pt}{4.818pt}}
\put(366.0,131.0){\rule[-0.200pt]{0.400pt}{4.818pt}}
\put(366,90){\makebox(0,0){$2$}}
\put(366.0,839.0){\rule[-0.200pt]{0.400pt}{4.818pt}}
\put(580.0,131.0){\rule[-0.200pt]{0.400pt}{4.818pt}}
\put(580,90){\makebox(0,0){$4$}}
\put(580.0,839.0){\rule[-0.200pt]{0.400pt}{4.818pt}}
\put(795.0,131.0){\rule[-0.200pt]{0.400pt}{4.818pt}}
\put(795,90){\makebox(0,0){$6$}}
\put(795.0,839.0){\rule[-0.200pt]{0.400pt}{4.818pt}}
\put(1010.0,131.0){\rule[-0.200pt]{0.400pt}{4.818pt}}
\put(1010,90){\makebox(0,0){$8$}}
\put(1010.0,839.0){\rule[-0.200pt]{0.400pt}{4.818pt}}
\put(1224.0,131.0){\rule[-0.200pt]{0.400pt}{4.818pt}}
\put(1224,90){\makebox(0,0){$10$}}
\put(1224.0,839.0){\rule[-0.200pt]{0.400pt}{4.818pt}}
\put(1439.0,131.0){\rule[-0.200pt]{0.400pt}{4.818pt}}
\put(1439,90){\makebox(0,0){$12$}}
\put(1439.0,839.0){\rule[-0.200pt]{0.400pt}{4.818pt}}
\put(151.0,131.0){\rule[-0.200pt]{0.400pt}{175.375pt}}
\put(151.0,131.0){\rule[-0.200pt]{310.279pt}{0.400pt}}
\put(1439.0,131.0){\rule[-0.200pt]{0.400pt}{175.375pt}}
\put(151.0,859.0){\rule[-0.200pt]{310.279pt}{0.400pt}}
\put(0,300){\rotatebox{90}{Peak memory [GB]}}
\put(795,29){\makebox(0,0){Number $k$ of blocks}}
\put(1279,818){\makebox(0,0)[r]{$\ell=10$}}
\put(1299.0,818.0){\rule[-0.200pt]{24.090pt}{0.400pt}}
\put(366,802){\usebox{\plotpoint}}
\multiput(366.58,798.56)(0.500,-0.912){425}{\rule{0.120pt}{0.829pt}}
\multiput(365.17,800.28)(214.000,-388.279){2}{\rule{0.400pt}{0.414pt}}
\multiput(580.00,410.92)(0.797,-0.499){267}{\rule{0.737pt}{0.120pt}}
\multiput(580.00,411.17)(213.470,-135.000){2}{\rule{0.369pt}{0.400pt}}
\multiput(795.00,275.92)(1.563,-0.499){135}{\rule{1.346pt}{0.120pt}}
\multiput(795.00,276.17)(212.206,-69.000){2}{\rule{0.673pt}{0.400pt}}
\multiput(1010.00,206.92)(2.627,-0.498){79}{\rule{2.188pt}{0.120pt}}
\multiput(1010.00,207.17)(209.459,-41.000){2}{\rule{1.094pt}{0.400pt}}
\put(366,802){\makebox(0,0){$+$}}
\put(580,412){\makebox(0,0){$+$}}
\put(795,277){\makebox(0,0){$+$}}
\put(1010,208){\makebox(0,0){$+$}}
\put(1224,167){\makebox(0,0){$+$}}
\put(1349,818){\makebox(0,0){$+$}}
\put(1279,777){\makebox(0,0)[r]{$\ell=11$}}
\multiput(1299,777)(20.756,0.000){5}{\usebox{\plotpoint}}
\put(1399,777){\usebox{\plotpoint}}
\put(366,804){\usebox{\plotpoint}}
\multiput(366,804)(9.945,-18.218){22}{\usebox{\plotpoint}}
\multiput(580,412)(17.578,-11.037){12}{\usebox{\plotpoint}}
\multiput(795,277)(19.763,-6.342){11}{\usebox{\plotpoint}}
\multiput(1010,208)(20.402,-3.813){11}{\usebox{\plotpoint}}
\put(1224,168){\usebox{\plotpoint}}
\put(366,804){\makebox(0,0){$\times$}}
\put(580,412){\makebox(0,0){$\times$}}
\put(795,277){\makebox(0,0){$\times$}}
\put(1010,208){\makebox(0,0){$\times$}}
\put(1224,168){\makebox(0,0){$\times$}}
\put(1349,777){\makebox(0,0){$\times$}}
\sbox{\plotpoint}{\rule[-0.400pt]{0.800pt}{0.800pt}}%
\sbox{\plotpoint}{\rule[-0.200pt]{0.400pt}{0.400pt}}%
\put(1279,736){\makebox(0,0)[r]{$\ell=12$}}
\sbox{\plotpoint}{\rule[-0.400pt]{0.800pt}{0.800pt}}%
\put(1299.0,736.0){\rule[-0.400pt]{24.090pt}{0.800pt}}
\put(366,805){\usebox{\plotpoint}}
\multiput(367.41,798.10)(0.500,-0.915){421}{\rule{0.121pt}{1.662pt}}
\multiput(364.34,801.55)(214.000,-387.551){2}{\rule{0.800pt}{0.831pt}}
\multiput(580.00,412.09)(0.803,-0.501){261}{\rule{1.484pt}{0.121pt}}
\multiput(580.00,412.34)(211.921,-134.000){2}{\rule{0.742pt}{0.800pt}}
\multiput(795.00,278.09)(1.591,-0.501){129}{\rule{2.729pt}{0.121pt}}
\multiput(795.00,278.34)(209.335,-68.000){2}{\rule{1.365pt}{0.800pt}}
\multiput(1010.00,210.09)(2.646,-0.502){75}{\rule{4.376pt}{0.121pt}}
\multiput(1010.00,210.34)(204.918,-41.000){2}{\rule{2.188pt}{0.800pt}}
\put(366,805){\makebox(0,0){$\ast$}}
\put(580,414){\makebox(0,0){$\ast$}}
\put(795,280){\makebox(0,0){$\ast$}}
\put(1010,212){\makebox(0,0){$\ast$}}
\put(1224,171){\makebox(0,0){$\ast$}}
\put(1349,736){\makebox(0,0){$\ast$}}
\sbox{\plotpoint}{\rule[-0.200pt]{0.400pt}{0.400pt}}%
\put(151.0,131.0){\rule[-0.200pt]{0.400pt}{175.375pt}}
\put(151.0,131.0){\rule[-0.200pt]{310.279pt}{0.400pt}}
\put(1439.0,131.0){\rule[-0.200pt]{0.400pt}{175.375pt}}
\put(151.0,859.0){\rule[-0.200pt]{310.279pt}{0.400pt}}
\end{picture}